\documentclass[]{interact}

\usepackage{epstopdf}
\usepackage[caption=false]{subfig}
\usepackage{listings}
\usepackage[longnamesfirst,sort]{natbib}
\bibpunct[, ]{(}{)}{;}{a}{,}{,}

\theoremstyle{plain}
\newtheorem{theorem}{Theorem}[section]

\theoremstyle{definition}
\newtheorem{definition}[theorem]{Definition}
\newtheorem{example}[theorem]{Example}

\theoremstyle{remark}

 \usepackage{floatpag}
 
\begin{document}

\articletype{ARTICLE TEMPLATE}

\title{Multi-output, multi-level, multi-gate design using non-linear programming}
%

\author{
\name{A.~C. Dimopoulos\textsuperscript{a}\thanks{CONTACT A.~C. Dimopoulos. Email: adimopoulos@hna.gr}, C.~Pavlatos\textsuperscript{b}  and G. Papakonstantinou\textsuperscript{c}}
\affil{\textsuperscript{a}Hellenic Naval Academy, Piraeus, Greece; \textsuperscript{b}Hellenic Air Force Academy, Dekelia Air base, Greece; \textsuperscript{c}National Technical University of Athens, Athens,Greece }
}

\maketitle

\footnotetext{A. C. Dimopoulos and C. Pavlatos have equally contributed to this work}
\begin{abstract}
 Using logic gates is the traditional way of designing logic circuits. However, most of the minimization algorithms concern a limited set of gates (complete sets), like sum of products,  exclusive-or sum of products, NAND gates, NOR gates e.t.c.. In this paper, a method is proposed for minimizing multi-output Boolean functions using any kind of two-input gates although it can easily be extended to multi-input gates. The method is based on non-linear mixed integer programming. The experimental results show that the method gives the same or better results compared to other methods available in the literature. However, other methods do not ensure that they produce the minimal solution, while the main advantages of the proposed method are that it does guarantee minimality and it can also handle Boolean functions for incompletely specified functions. The method is general enough and can easily be extended to more complicated design modules than just basic gates.
\end{abstract}


\section{Introduction}
The conventional method of designing a logic circuit is via using logic gates. Nonetheless, the vast majority of the minimization algorithms involve a limited set of gates (complete sets), e.g. sum of products (SOPs),  Exclusive-Or-Sum-Of-Products (ESOPs), NAND gates, NOR gates etc. In the literature there is a very limited number of publications that permit the combination of any kind of gates, such as Refs.\cite{sasao}, \cite{mehdi}, \cite{karak}, \cite{arez}, \cite{Coello}. Moreover, most of them are using genetic algorithms which cannot guarantee minimality. However, in previous approaches that guarantee minimality, presented by our team, methods were proposed that can also handle incompletely specified Boolean functions by using as modules Multiplexers (MUXs) or Reed–Muller universal blocks (RMs) \cite{NonLinear} and NAND gates \cite{a1}. Nevertheless, the methodology presented in references \cite{NonLinear,a1} was restricted to specific gates. In this paper a method is proposed for minimizing multi-output Boolean functions using any kind of two-input gates (Fig. \ref{fig:Ex1}), although it can easily be extended to multi-input gates. The method allows the user to select the kind of gates to be used e.g. AND, OR, NAND, NOR, XOR, NOT, etc., as well as the architecture of the desired circuit.
\begin{figure}[!t]
\caption{General form of a gate}	
	\centering
		\includegraphics [scale=0.80] {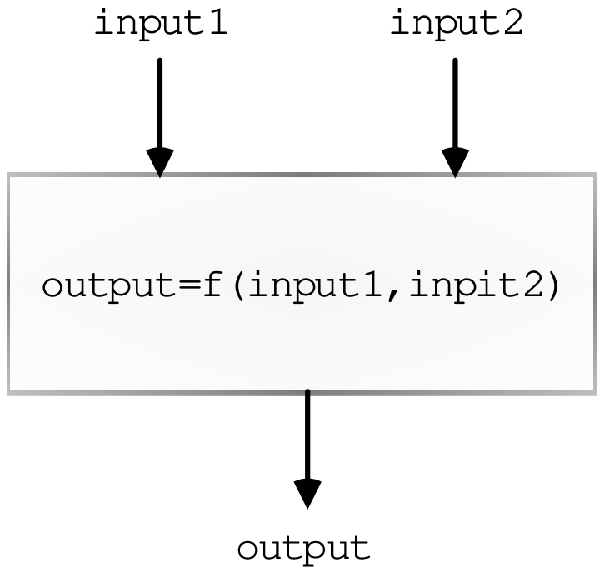}  
	
	\label{fig:Ex1}
\end{figure}
This architecture can have for example the structure of an $p \times q$ grid of gates for at most $q$  functions (multi-output function) or a tree structure (Fig. \ref{fig:Ex2}) for a single-output function. The inputs of each gate can be connected to the outputs of the gates at any previous level of gates in the grid or they can be variables or constants. \\The method is based on non-linear mixed integer optimization methods, with the goal of minimizing the number of gates used in the structure or the number of transistors within the circuits. For this goal, any suitable optimizer can be utilized. In our case, we choose to run all the experimental examples on the free access NEOS server \cite{neos1,neos2,neos3,neos4} and largely the BARON solver \cite{baron}. BARON implements deterministic global algorithms of the branch-and-bound type, that guarantee global optima under fairly general assumptions. These assumptions are fulfilled in our case. Therefore we can locate an exact solution, which provided the minimum number of gates. For facilitating the creation of the required by the NEOS server GAMS nomenclature, a FORTRAN program has been implemented for the automatic translation of the problem description to GAMS. The experimental results show that the method gives the same or better results compared to other methods available in the literature. However, other methods do not ensure that they produce the minimal solution, while the main advantages of the proposed method are that it does guarantee minimality and it can also handle Boolean functions for incompletely specified functions. The method is general enough and can easily be extended to basic design modules more complicated than gates.\\
The remainder of the paper is organized as follows. In Section \ref{sec:basicd} some definitions and preliminaries are provided, to satisfy the purpose of a self-contained paper and for reader's ease. In Section \ref{sec:trans} the basic idea is described,  while in Section \ref{sec:illustr} an illustrative example is given. In Section \ref{sec:donotcare} the case of incompletely specified functions is described. Experimental results are presented for several numerical examples in Section \ref{sec:implement}. Finally, Section \ref{sec:conclusions} concludes this work.

\section{Preliminaries}\label{sec:basicd}
In this section, important concepts and definitions from the area of Boolean algebra are presented, in order to build up a necessary background for the rest of the paper.    
\begin{definition} 
A Boolean function $f$ is a mapping $f:{{\left\{ 0,1 \right\}}^{n}}\to \left\{ 0,1 \right\}$.
\end{definition}
\begin{definition} 
Let $x$ be a variable that takes a value from $V=\{0,1\}$ and $S\subseteq V$. Then $x^S$  is a literal of $x$, such that $x^S=1$ when $x\in S$ and $x^S=0$ when $x\in V\char92 S$. When $S=V$ then $x^S=1$. 
\end{definition} \noindent
A common notation denotes $x$ for $x^{\left\{1\right\}}$, $\bar{x}$ for $x^{\left\{0\right\}}$ and 1 for $x^{\left\{0,1\right\}}$.
\begin{definition} 
If $\dot{x}_j$, with $1\leq j \leq n$, is a literal of the variable $x_j$ then the expression $C=\dot{x}_1\dot{x}_2$...$\dot{x}_n$ is a product term or cube. When $\dot{x}_j=x_j$ or $\bar{x}_j$ (excluding $\dot{x}_j=0,1$), then $C$ is called minterm and we denote $\dot{x}_j$ as $\ddot{x}_j$ for all $j$'s.
\end{definition}\noindent
Overall, there exist $2^n$ minterms.
\begin{definition}
If we replace  each $x_j$ with 1 and each $\bar{x}_j$  with 0 in a minterm, we form a binary number $g$ that represents the specific minterm, which is called representative number. 
\end{definition}
\begin{definition}
For each minterm with representative number $g$ and for each variable $j$ in the representation of the minterm, we define the representative bit $b_{g,j}$, $1\leq j \leq n$, as:
\begin{itemize}
\item $b_{g,j}=1$ if the variable is in its normal form i.e. $x_j$
\item $b_{g,j}=0$ if the variable is in its negated form i.e. $\bar{x}_j$
\end{itemize}
\end{definition}\noindent
Each Boolean function $f$ can be uniquely represented as the Boolean sum of all minterms for which $f(\ddot{x}_1,\ddot{x}_2,...,\ddot{x}_n)=1$ (minterm expression of the function).
\begin{example}
Let us consider the minterm $x_1\bar{x_2}x_3\bar{x}_4$. Its representative number is $1010_{(2)}$ in binary form or $10_{(10)}$ in decimal form. Hence, this is the 10\textsuperscript{th} minterm i.e. $g=10$. Its representative bits are:

$b_{12,1}=1$

$b_{12,2}=0$

$b_{12,3}=1$

$b_{12,4}=0$
\end{example}

\begin{definition}\label{def6}  
We say that a product term covers a minterm, if it is 1 when the minterm is 1.
\end{definition} \noindent
This happens when the product term has the constant 1 or the same literal with the minterm for all corresponding variables. It is noted that a missing variable in the product term can be considered as 1.
\begin{definition}
The bitvector representation of a Boolean function of $n$ variables is a $2^n$ bit vector, where the $g$\textsuperscript{th} bit {\normalfont(}$0\leq g \leq 2^n-1${\normalfont)} is 1 if the minterm with representative number $g$ is included in the minterm expression of the function, otherwise the $g$\textsuperscript{th} bit is 0.
\end{definition}
%
\begin{example}
The functions $f_1=x_1x_3$ and $f_2=x_1\bar{x}_2x_3$ can be represented as a Boolean sum of minterms, e.g. $f_1=x_1x_3=x_1(x_2+\bar{x}_2)x_3=x_1x_2x_3+x_1\bar{x}_2x_3$ and $f_2=x_1\bar{x}_2x_3$. The representative numbers of the minterms $x_1x_2x_3$ and $x_1\bar{x}_2x_3$ are 7 and 5, respectively. Hence, the bitvector form of function $f_1$ is $f_1=10100000$ and of function $f_2$ is $f_2=00100000$, with the least significant bit the rightmost one and the most significant bit the leftmost one. We also note that according to definition \ref{def6} the product term $x_1x_3=x_11x_3=$ covers the minterm $x_1\bar{x}_2x_3$ .
\end{example}
\begin{definition}
An Exclusive-or Sum of Product terms (ESOP) is an expression of the form $\oplus\sum_{i=1}^{m}C_i$, where $C_i$ are cubes, that non-uniquely represents a function and $\oplus$ the XOR boolean function. If $C_i$ are minterms, then the expression uniquely represents the function.
\end{definition}\noindent
It is easy to see that when the values \{0,1\} of the variables are given, only one minterm will be 1 and all others will be 0. Hence, it is indifferent if we have Boolean sum or XOR sum of the same minterms and each function $f$ can be also uniquely represented as a XOR sum of its minterms $\ddot{x}_1\ddot{x}_2...\ddot{x}_n$.
\begin{example}
Let us consider the function $f=11001010$ in bitvector form or $f=x_1x_2x_3$$\oplus$$x_1x_2\bar{x}_3\oplus\bar{x}_1x_2x_3\oplus\bar{x}_1\bar{x}_2x_3$ in minterm form. The expression $x_1x_2\oplus\bar{x}_1x_3$  is another ESOP expression of the same function $f$ with only two product terms, which can easily be verified.
\end{example}
\begin{definition} 

Let $f(\textbf{x})$ be a switching
function and \textbf{x} the vector of its variables. Let $x_i$ be
one of the variables in the vector \textbf{x}. Then
$f(x_1,x_2,\ldots,x_i=0,\ldots,x_n)$, $f(x_1,x_2,\ldots,x_i=1,\ldots,x_n)$, \allowbreak 
$\{f(x_1,x_2,\ldots,$\allowbreak $x_i=0,\ldots,x_n) \oplus 
f(x_1,x_2,$\allowbreak $\ldots,x_i=1,\ldots,x_n)\}$
are subfunctions of $f$,
regarding variable $x_i$. For simplicity, in the rest of this
paper, they will be referred as $f^0$, $f^1$ and $f^2$
respectively.
\end{definition}\noindent
A Boolean function $f$ can thus be expressed as:
\begin{equation}
\label{eq:eq70}
\begin{array}{c}
f(\textbf{x})=\bar{x}_{n}f^{0}\oplus{x}_{n}f^{1}=
{x}_{n}f^{2}\oplus{f}^{0}=
\bar{x}_{n}f^{2}\oplus{f}^{1}
\end{array}
\end{equation}
These expressions are called Shannon, Positive Davio and Negative Davio expansions respectively. The Shannon expansion is also known more frequently in the equivalent form:
\begin{equation}
\label{eq:eq71}
\begin{array}{c}
f(\textbf{x})=\bar{x}_{n}f^{0} + {x}_{n}f^{1}
\end{array}
\end{equation}
\begin{theorem}
Every Boolean function can be implemented using $3^n$gates at most, where $n$ is the number of variables.
\end{theorem}
\begin{proof}
Let us consider the tree topology of	Fig. \ref{fig:Ex2} with gates as nodes. A simple solution would be to start with the given function at the root of the tree, and apply recursively the Shannon expansion or the Positive (Negative) Davio expansion respectively up to the inputs of the leave modules. These inputs will be constants 0 or 1, depending on the minterms of the function, which are known.
\end{proof} \noindent
Obviously the implementation described previously is not optimal. The problem is to find a solution with the least number of different kinds of gates and the smaller number of levels (Fig. \ref{fig:Ex2}). This will result in power efficiency and delay reduction.
\begin{figure}[!t]
	\caption{The gate tree}
	\centering
		\includegraphics [scale=0.55] {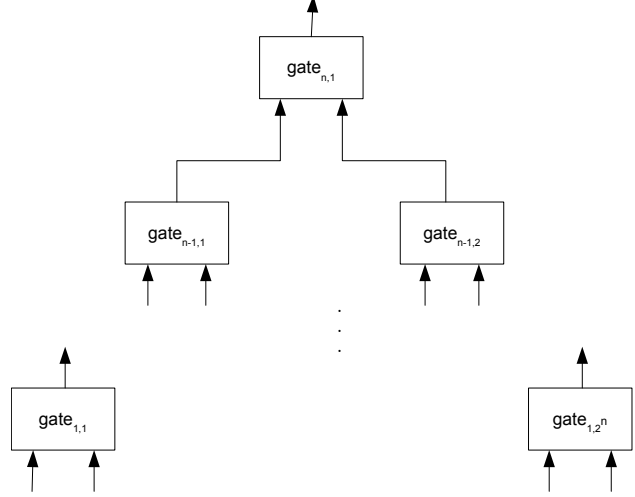}  

	\label{fig:Ex2}
\end{figure}
\section{The basic idea}\label{sec:trans}
In this Section an overview of the basic idea is given, where the input of each gate can either be binary variables $x_i$ of a Boolean function, or binary variables that are the output of previous level gates, or a constant 0 or 1. Each gate's type, input, and output corresponds to an unknown variable to be computed. Hence, for every minterm $g$, the following relations should be written:
\begin{itemize}
\item
Every gate type $t$ is expressed by the following expression\\
$tsel_{i1} \times f_1 + tsel_{i2} \times f_2 + tsel_{i3} \times f_3 + tsel_{i4} \times f_4 + tsel_{i5} \times f_5 + tsel_{i6} \times f_i6 + tsel_{i7} \times f_7$,\\ 
where $tsel_ij$ are binary unknown variables to select a gate type and the expression sum should be 1. The $f_j$s are defined as:

$f_1=AND(x,y) \Rightarrow x \times y$.

$f_2=OR(x,y) \Rightarrow x + y-(x \times y)$.

$f_3=NOT(x) \Rightarrow (1 - x)$.

$f_4=XOR(x,y) \Rightarrow x + y - 2 \times (x \times y)$.

$f_5=NAND(x,y) \Rightarrow (1 - x) + (1-y) - (1 - x) \times (1 - y)) = 1 - x \times y$.

$f_6=NOR(x,y) \Rightarrow (1-x) \times (1 - y)$.

$f_7=CON(x) \Rightarrow x$.\\ 
It is noted that  the above expressions show the obvious equivalency between  Boolean algebra expressions and ordinary algebra expressions. It is noted that the ``gate'' $CON$ is actually a wire connection. Moreover, since each gate should be only of one type the following constrain should hold for the gate selection binary variables $tsel_{ij}$

$tsel_{i1} + tsel_{i2} + tsel_{i3} + tsel_{i4} + tsel_{i5} + tsel_{i6} + tsel_{i7} = 1$. 
\item
The two inputs of each gate are fed by one of a set of possible entities $T_j$. These entities can be outputs of gates of the previous level, the variables of the function or the constants 0, 1. For each input $inp_i$ we write a relation of the form:

$inp_i=inpsel_{i1} \times T_1 + inpsel_{i2} \times T_2 +inpsel_{i3} \times T_3 + \ldots$\\
where $inpsel_{ij}$ is a binary selection coefficient corresponding to each entity. It means that the corresponding entity $T_{j}$ will be the only one to feed the input $inp_i$ if $inpsel_{ij}=1$ or not if $inpsel_{ij}=0$.\\
Since only one entity can feed an input, the following constrain should be used for each of the inputs:

$inpsel_{i1}+inpsel_{i2}+inpsel_{i3}+ \ldots =1$.
\item
The output of each gate, which is a function $f_i$ of its data inputs and its gate type (as described before), is written as $out_i=E_i \times f_i$. The coefficient $E_i$ which is also a binary variable to be computed, signifies that the corresponding gate will be active if its value is 1 or inactive if its value is 0. If it is inactive it means that the corresponding module can be eliminated.\\
Since, the target is to reduce the number of gates used as much as possible, we must minimize the expression:
$E_1+E_2+E_3+ \ldots$. It is noted that the ``gates'' corresponding to the operation CON are not taken into account in the previous sum, since it corresponds to a simple wire connection.
\item
Finally we have to ensure that the output(s) of the circuit (the output(s)  at the top level), will produce the given function(s) $F$. Hence, we have to write the relation(s):\\
\begin{itemize}
    \item $O_n=1$, if the examined minterm is covered by the function $f$ or
    \item $O_n=0$, otherwise
\end{itemize}
\end{itemize}
All the above described relations, for each minterm, constitute the integer non-linear problem to be solved. The next illustrative example will clarify the proposed non-linear integer programming approach.
\section{An illustrative example}\label{sec:illustr}
Let us consider the implementation of the three variable function $f=\sum(0,1,3,$ $5,6)=01101011=6b$ . This function has five minterms with representative numbers 0, 1, 3, 5, 6 and  representative bits:\\
\begin{itemize}
\item []
$b_{0,1}=0$ for minterm $0_{10}=(000)_2$ and variable 1
\item []
$b_{0,2}=0$ for minterm $0_{10}=(000)_2$ and variable 2 
\item []
$b_{0,3}=0$ for minterm $0_{10}=(000)_2$ and variable 3
\item []
$b_{1,1}=0$ for minterm $1_{10}=(001)_2$ and variable 1
\item []
$b_{1,2}=0$ for minterm $1_{10}=(001)_2$ and variable 2
\item []
$b_{1,3}=1$ for minterm $1_{10}=(001)_2$ and variable 3
\item []
$b_{2,1}=0$ for minterm $2_{10}=(010)_2$ and variable 1
\item []
$b_{2,2}=1$ for minterm $2_{10}=(010)_2$ and variable 2
\item []
$b_{2,3}=0$ for minterm $2_{10}=(010)_2$ and variable 3
\item [] \dots 
\item [] \dots 
\item [] \dots 
\item []
$b_{7,3}=1$ for minterm $7_{10}=(111)_2$ and variable 3
\end{itemize}
For this example we examine gates that can take as input:
\begin{itemize}
    \item the output of the exactly previous level (not true for the first level)
    \item either constant 0 or 1
    \item the variable of the function
\end{itemize}
These gates can be placed on a grid formation of $3 \times 2$ to implement the given function. Considering the gate $(i,j)$ we can construct the following equations for each of its two inputs $inp1_{i,j,g}$ and $inp2_{i,j,g}$, its output $out_{i,j,g}$ and for each minterm $g$ of the function to be implemented:\\
$inp1_{i,j,g}=inpsel1_{i,j,1} \times out_{i-1,1,g} + inpsel1_{i,j,2} \times out_{i-1,2,g} + inpsel1_{i,j,3} \times b_{g,1} +inpsel1_{i,j,4} \times b_{g,2} +inpsel1_{i,j,5} \times b_{g,3} + inpsel1_{i,j,6} \times 1 + inpsel1_{i,j,7} \times 0$\\
The above expression signifies that gate $(i,j)$ takes as first input either:
\begin{itemize}
    \item one of the outputs of the two gates of the previous level or
    \item one of the three variables, for which the minterm $g$ will have the values $b_{g,1},b_{g,2},b_{g,3}$ respectively or
    \item one of the constants 0 or 1
\end{itemize}
Clearly, for all the $b's$ that are equal to zero in the given function, the corresponding term in the above expression can be eliminated. 
The choice between all entities that will feed the input is made by the selection parameters $inpsel_ {i,j,k}$, hence only one of these can be 1 while all the rest have to be 0. Mathematically this can be expressed by the following relation:\\
$inpsel1_ {i,j,1} + inpsel1_ {i,j,2} + inpsel1_ {i,j,3} + inpsel1_ {i,j,4} + inpsel1_ {i,j,5} + inpsel1_ {i,j,6} + inpsel1_ {i,j,7} = 1$  
\\
It is noted that all the above variables represent binary ones. We can write corresponding equations for the input $inp2_{i,j,g}$\\
$inp2_{i,j,g}=inpsel2_{i,j,1} \times out_{i-1,1,g} + inpsel2_{i,j,2} \times out_{i-1,2,g} + inpsel2_{i,j,3} \times b_{g,1} +inpsel2_{i,j,4} \times b_{g,2} +inpsel2_{i,j,5} \times b_{g,3} + inpsel2_{i,j,6} \times 1 + inpsel2_{i,j,7} \times 0$\\
$inpsel2_ {i,j,1} + inpsel2_ {i,j,2} inpsel2_ {i,j,3} + inpsel2_ {i,j,4} + inpsel2_ {i,j,5} + inpsel2_ {i,j,6} + inpsel2_ {i,j,7} = 1$\\
Obviously, at the first level in the expressions for the inputs the first two terms should be eliminated, since there is no previous level to provide outputs.\\
As far as the output $out_{i,j,g}$, we must ensure that the gate $(i,j)$:
\begin{itemize}
\item will be one of a given set of gate types (described in the previous Section)
 \item can be (potentially) eliminated
\item will produce at the top level of the examined architecture the desired function
\end{itemize}\noindent
Hence, we can write the following equation for each minterm $g$:\\
$out_{i,j,g}= E_{i,j} \times (outsel_{i,j,1} \times (inp1_{i,j,g} \times inp2_{i,j,g}) + outsel_{i,j,2} \times (inp1_{i,j,g} + inp2_{i,j,g} - inp1_{i,j,g} \times inp2_{i,j,g}) + outsel_{i,j,3} \times (1 - inp1_{i,j,g}) + outsel_{i,j,4} \times (inp1_{i,j,g} + inp2_{i,j,g} - 2 \times inp1_{i,j,g} \times inp2_{i,j,g}) + outsel_{i,j,5} \times ((1-inp1_{i,j,g}) + (1-inp2_{i,j,g}) - 2 \times (1-inp1_{i,j,g}) \times (1-inp2_{i,j,g})) + outsel_{i,j,6} \times ((1-inp1_{i,j,g}) \times (1-inp2_{i,j,g})) + outsel_{i,j,7} \times (inp1_{i,j,g}))$\\
The above equation describes that a gate can only be one of the possible functions $f_k$, described in the previous Section, determined by the selection variables $outsel$. It is noted that we can use a desired subset of the above gates, omitting the appropriate lines in the above equation. Moreover, variable $E_{i,j}$ defines if the gate $(i,j)$ will be active ($E_{i,j}=1$) or inactine ($E_{i,j}=0$) in which case the gate can be eliminated. Hence, two more equations should be added for each minterm $g$, in order to ensure the above requirements:\\
$outsel_{i,j,1} + outsel_{i,j,2}  
+ outsel_{i,j,3} + outsel_{i,j,4} +
 + outsel_{i,j,5} + outsel_{i,j,6} + outsel_{i,j,7} =1$\\
%
$Obj=E_{1,1} \times (1-outsel_{1,1,7}) + E_{1,2} \times (1-outsel_{1,2,7}) + E_{2,1} \times (1-outsel_{2,1,7}) + E_{2,2} \times (1-outsel_{2,2,7}) + E_{3,1} \times (1-outsel_{3,1,7}) + E_{3,2} \times (1-outsel_{3,2,7}) )$
\begin{figure}[t]
	\caption{Circuit of function $f=6b$}
	\centering
		\includegraphics [scale=0.55] {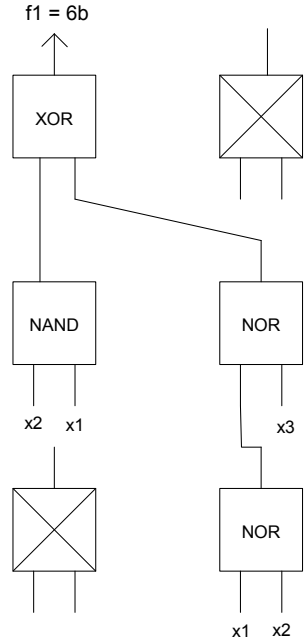}  

	\label{fig:Ex3}
\end{figure}
\\$Obj$ in the above equation is the objective function to be minimized, so that the produced circuit will have the minimum number of gates. $Obj$ is an integer variable and is equal to the sum of all $E_{i,j}$s, which are binary variables indicating that a gate is active.  It is multiplied by $(1-outsel_{i,j,7})$ for each gate $(i,j)$. This is because in case the gate $(i,j)$ is a simple wire connection, it is not counted in the cost of the circuit.\\
Finally, we have to ensure that the produced circuit will have as output the given function. Hence, the following equations should be added for each minterm $g$:
$out_{3,1,g} = 1$ or $out_{3,1,g} = 0$, depending on whether function f covers or not the minterm g.\\
The search for a solution based on all the above equations establish the non-linear problem for the specific illustrative example. The final circuit that results from the that solution is shown in Fig. \ref{fig:Ex3} and consists of four gates. \\
In case we had a two-output function with outputs $6b$ and $2a$, we would have to add the corresponding expressions for $out_{3,2,g}$ i.e. for the top level (output) gate $3,2$. This final circuit is shown in Fig. \ref{fig:Ex4}, which requires five gates.
\begin{figure}[h]
	\caption{Circuit of functions $f_1=6b$ and $f_2=2a$ }
	\centering
		\includegraphics [scale=0.55] {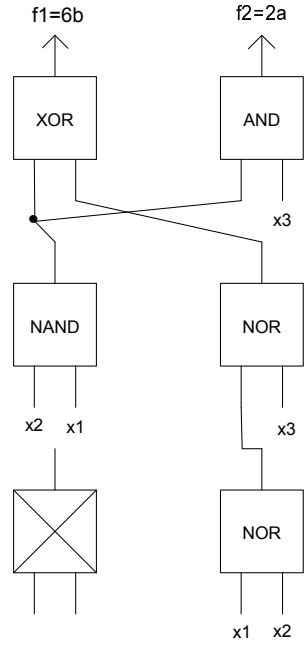}  
	\label{fig:Ex4}
\end{figure} 

\section{Incompletely specified functions}\label{sec:donotcare}
The proposed method can also tackle with the more difficult problem of minimizing expressions while taking into consideration do-not-care conditions. The do-not-care conditions refer to cases where for given minterms we do-not-care if their value in the bitvector form of the function will be 0 or 1. These functions are called incompletely specified functions and one way of describing them is through two disjoint sets:
\begin{itemize}
    \item \textit{on-}set, which includes all the minterms for which the function becomes 1
    \item do-not-care set (\textit{dc-}set), which includes all the do-not-care minterms
\end{itemize}
Those minterms that do not belong to either set, are the ones for which the function becomes 0 and comprise the \textit{off}-set. Representing these sets as functions in their bitvector form results equivalently in the \textit{on-} and \textit{dc-} functions.\\
In these cases with incompletely specified functions, all expressions related to minterms corresponding to the do-not-care ones are skipped.
For example, if the minterms with representative numbers 0 and 1 are do-not-care ones, i.e. $\bar{x_1}\bar{x_2}\bar{x_3}$ and $\bar{x_1}\bar{x_2}{x_3}$, then all equations with $g=0$ or $g=1$ are skipped, e.g. $output_{i,j,1}=1$ and all others with subscript $g=1$ in the previous illustrative example. The solution obtained in this case is shown in Fig.\ref{fig:Ex5} for the case of the two functions $6b$ and $2a$, a circuit requiring three gates.

\begin{figure}[t]

	\caption{Circuit of functions $f_1=6b$ , $f_2=2a$ and DCs }
	\centering
		\includegraphics [scale=0.55] {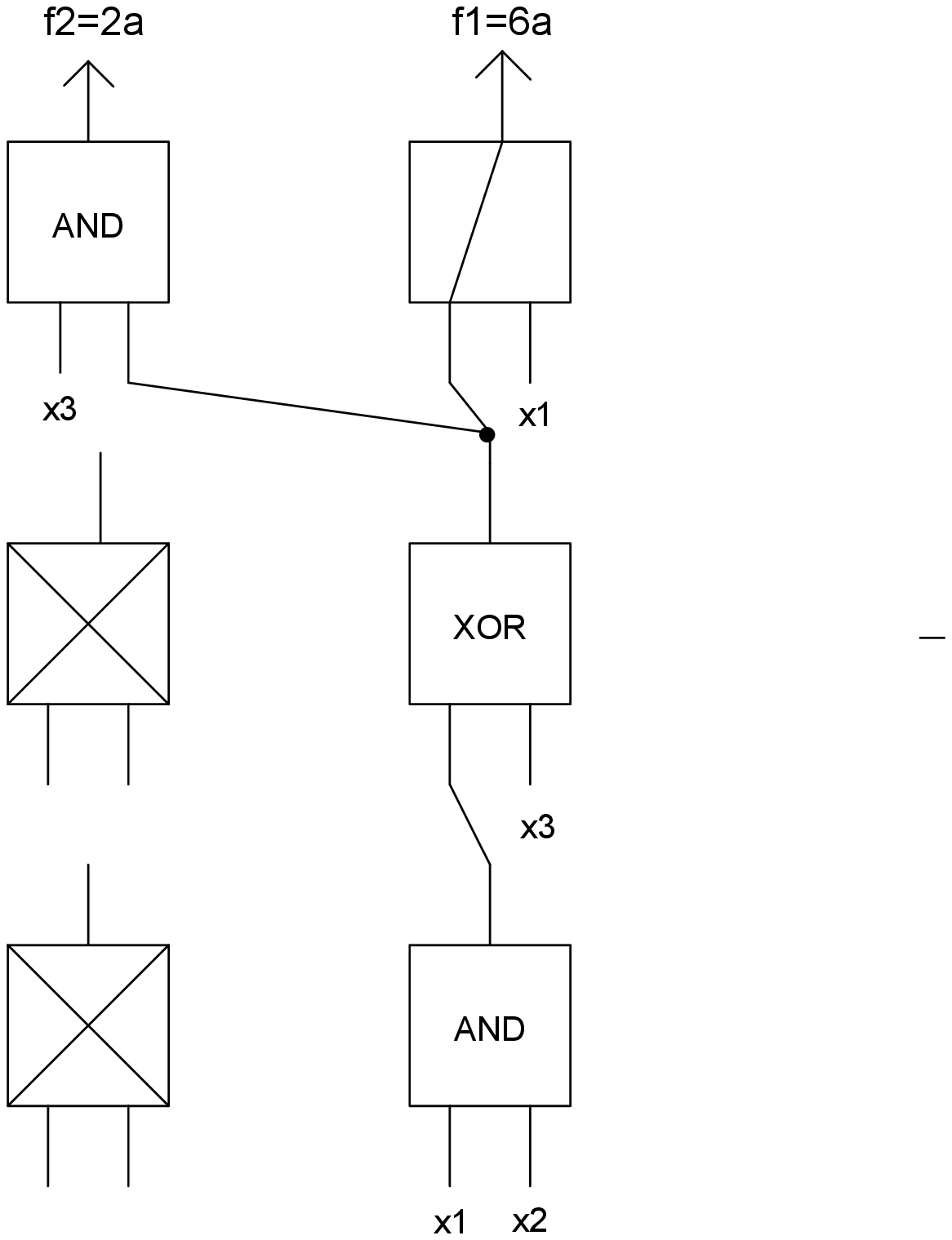}  
	\label{fig:Ex5}
\end{figure}  

\section{Implementation}\label{sec:implement}

All our experiments were executed on the free access NEOS server \cite{neos1,neos2,neos3,neos4}. We mainly used the BARON solver\cite{baron}, which implements deterministic global algorithms of the branch-and-bound type, as a mixed-integer optimal constrained optimizer.
Hence, BARON solutions are guaranteed to provide global optima under fairly general assumptions. These assumptions are fulfilled in our case, hence we can find an exact expression of a function.\\
The free access service of NEOS limits the maximum computing time to 8 hours, which was not enough for  some of our examples to run in full. Hence, for these cases the final solution was not found, but instead the best one found so far. \\
The nomenclature used by the NEOS server is that of the GAMS and AMPL formal languages. In order to ease the production of the required GAMS equivalent programs, a program in FORTRAN was implemented to automatically produce the required GAMS equivalent programs for a given number of variables, number of levels, number of gates at each level, type of allowable gates and the given function to be implemented.\\
We tested all the examples of references \cite{mehdi}, \cite{karak}, \cite{arez} using two-input gates and in all cases the results obtained by our approach were at least the same or even better in some cases. All the results of these comparisons are summarized in Table \ref{tab:1}, where those examples that exhausted the maximum allowed running time of the NEOS system are indicated with a  star character, e.g. example functions $0ee9$, $5a5a$, $936c$. However, even for such cases without the guarantee for optimal solutions, the results were the same or better. For those examples, where no star character is shown, the optimizer terminated within the maximum provided time and hence produced the optimal solution, e.g. example functions $a7f1$, $ab$, $4a6a$. \\
The GAMS program for the example of Fig.\ref{fig:Ex5} is given in the Appendix. As it can be observed it is not a difficult task to transform the non-linear-program  to a GAMS equivalent one.

\begin{center}

\begin{table}[p]
\thisfloatpagestyle{empty} 
\caption{Experimental Results}
\label{tab:1}
{
\begin{tabular}{|p{0.3\textwidth}|p{0.25\textwidth}|p{0.2\textwidth}|p{0.2\textwidth}|}
\hline

Example  &Results&Results&Allowable\\
Function&gates of & gates of &gates in\\
in HEX&Ref$^{x}$& our method&circuit\\
%

\hline \hline
$0ee9$&Ref.\cite{sasao}$\rightarrow$   11&$8^*$&All gates\\
$\sum(0,3,5,6,7,9,10,11)$&example 11.2&&\\
\hline
$a7f1$&Ref.\cite{mehdi}$\rightarrow$   5&$5$&All gates\\
\hline
$5a5a$&&&\\
$936c$&Ref.\cite{mehdi}$\rightarrow$   7&$7^*$&All gates\\
$ec80$&&&\\
\hline
$a0a0$&&&\\
$6ac0$&Ref.\cite{mehdi}$\rightarrow$   7&$7^*$&All gates\\
$4c00$&&&\\
$8000$&&&\\
\hline
&&&AND,OR,\\
$25cb$&Ref.\cite{karak}$\rightarrow$   7&$7^*$&XOR,NOT,\\
&&&CON\\
\hline
&&&AND,OR,\\
$a7f1$&Ref.\cite{karak}$\rightarrow$   7&$6^*$&XOR,NOT,\\
&&&CON\\
\hline
$ab$&Ref.\cite{arez}$\rightarrow$   5&5&NAND\\
$\sum(0,1,3,5,7)$&&&\\
\hline
$69$ &Ref.\cite{arez}$\rightarrow$   13&$12$&NAND\\
$\sum(0,3,5,6)$&&&\\
\hline
$4a6a$&Ref.\cite{arez}$\rightarrow$   9&9&NAND\\
$\sum(1,3,5,6,9,11,14)$&&&\\
\hline
$22d5$&Ref.\cite{arez}$\rightarrow$   9&8&NAND\\
$\sum(0,2,4,6,7,9,13)$&&&\\
\hline
$aaaaaaa8$&Ref.\cite{bara}$\rightarrow$8&5&NAND,\\
$\sum(3,5,7,..29,31)$&&& NOR\\
\hline
$96$&Ref.\cite{bhat}$\rightarrow >$12&$12^*$&NAND\\
$\sum(1,2,4,7)$&&&\\
\hline
$e8$&Ref.\cite{bhat}$\rightarrow>$6&6&NAND\\
$\sum(3,5,6,7)$&&&\\
\hline
$bafc$&Ref.\cite{bhat}$\rightarrow>$7&7&NAND\\
$\sum(2,3,4,5,6,7,9,$&&&\\
$11,12,13,15)$&&&\\
\hline
\end{tabular}
}
\end{table}
\end{center}

\section{Conclusions}\label{sec:conclusions}
In what was shown in the previous Sections, the proposed approach is a non-linear one that can be applied for designing multi-function, multi-level, two-input multi-gates logic circuits. Based on the presented experimental results, this method outperforms other methods available in the literature, while guaranteeing minimality. Moreover, it can tackle with  Boolean functions for incompletely specified functions, and it is flexible in defining the desired architecture to be used. Due to its generality, the method can be extended to use more complicated modules, instead of simple gates and to also support multi-input gates.\\
The presented experimental results of Section \ref{sec:implement} certify that large problems with hundreds or even thousands of unknown variables are manageable computational wise. The later is a great challenge for this type of non-linear integer programming problems, which the proposed method overcomes.\\
Our future endeavour will be to use this method for multi-input gates, as well as other more complicated modules e.g. for Exclusive Or Complex Terms \cite{voud} (ESCTs).

\section*{APPENDIX}
 
\appendix
\lstset{language=[90]Fortran,
  basicstyle=\ttfamily,
  commentstyle=\color{green},
  morecomment=[l]{!\ } Comment only with space after !
}
\begin{lstlisting}[caption={The FORTRAN code that implements the GAMS code for the example of Fig.\ref{fig:Ex5}}]


*-------------------------------------------------*
* GAMS program for orthogonal architecture iixjj, *
* multi-output functions and multi-type gates.    *
* Example of two functions 6b and 2a and two do-  *
* not-care minterms with g=0,1.                   *
* A star in the first column of a line inicates   *
* a comment line.                                 *
* The formula con(ii,jj,kk)..exp means, create    *
* rules for exp and for every possible value of   *
* the indices ii,jj,kk.                           *
* sum(m,exp)  means the sum of  all expressions   *
*  exp, over the index.                           *
* =e= means must be equal to.                     *
* z is the (dummy) objective function.            *
* obj..z=e=0; means ignore the satisfaction of    *
* the objective function and satisfy the          *
* constrains only.                                *
*-------------------------------------------------*
set ii no. of levels /1,2,3/;
set jj no. of gates at a level /1*2/;
set kk no. of inputs of a gate /1*2/;
set ll no. of constants /1*2/;
set gg no. of minterms /2*7/;
* Replace the previous definition of gg with  *
* the next one (in comment now)if we do not   *
* want to have the do-not-care minterms g=0,1.*
*set gg no. of minterms /0*7/;
set qq no. of variables /1*3/;
set rr no. of gate types /1*7/;
alias (pp,jj);
table f(gg,jj) output functions
   1  2
* Activate the next two data lines (deleting  *
* the star in the first column of the line),  *
* if we do not want to have the do-not-care   *
* minterms g=0,1.
*0  1  0
*1  1  1
2  0  0
3  1  1
4  0  0
5  1  1
6  1  0
7  0  0
;

table  t(gg,qq) auxiliary table
   1 2 3
* Activate the next two data lines (deleting  *
* the star in the first column of the line),  *
* if we do not want to have the do-not-care   *
* minterms g=0,1.
*0  0 0 0
*1  0 0 1
2  0 1 0
3  0 1 1
4  1 0 0
5  1 0 1
6  1 1 0
7  1 1 1
;

binary variable c(ii,jj) gate validation;
binary variable inp(ii,jj,kk,gg) inputs of gates;
binary variable out(ii,pp,gg) output of gates;
binary variable p(ii,jj,kk,pp) select output to feed the input;
binary variable q(ii,jj,kk,qq) select variable to feed the input;
binary variable r(ii,jj,rr) select gate type to feed the input;
binary variable con(ii,jj,kk,ll) select constant to feed the input;
Free Variable z objective;
equations con1,con2,con3,con4,con5,obj;

con1(ii,jj,kk)..sum(pp,p(ii,jj,kk,pp))+sum(qq, q(ii,jj,kk,qq))+
sum(ll,con(ii,jj,kk,ll))=e=1;
con2(ii,jj)..sum(rr,r(ii,jj,rr))=e=1;
con3(ii,jj,kk,gg)..inp(ii,jj,kk,gg)=e=sum(pp, p(ii,jj,kk,pp)
*out(ii-1,pp,gg))+sum(qq,q(ii,jj,kk,qq)*t(gg,qq))+con(ii,jj,kk,'1');
con4(ii,jj,gg)..out(ii,jj,gg)=e=c(ii,jj)*( r(ii,jj,'2')
*prod(kk,inp(ii,jj,kk,gg)) +r(ii,jj,'3')*(sum(kk,inp(ii,jj,kk,gg)) 
-prod(kk,inp(ii,jj,kk,gg))) +r(ii,jj,'6')*(1-inp(ii,jj,'1',gg)) 
+r(ii,jj,'4')*(sum(kk,inp(ii,jj,kk,gg)) -2*prod(kk,inp(ii,jj,kk,gg))) 
+r(ii,jj,'1')*(1 -prod(kk,inp(ii,jj,kk,gg))) +r(ii,jj,'5')
*prod(kk,(1-inp(ii,jj,kk,gg))) +r(ii,jj,'7')*inp(ii,jj,'1',gg));
con5(gg,jj)..out('3',jj,gg)=e=f(gg,jj);
obj..z=e=sum((ii,jj),(c(ii,jj)*(1-r(ii,jj,'7'))));
Model mplex /
con1,con2,con3,con4,con5,
Obj/;
Option MINLP = BARON;
Option threads=4;
mplex.reslim = 500;
Solve mplex using MINLP minimizing z;

\end{lstlisting}

\end{document}